%% file: deceptive_pls.tex
\documentclass[conference,a4paper,10pt]{IEEEtran}

\usepackage{silence}
\usepackage[left=1.62cm,right=1.6cm,top=1.9cm]{geometry}

\usepackage{cite}

\usepackage{amsmath,amssymb,amsthm,fixmath}

\usepackage{color,colortbl}
\usepackage{xcolor}

\usepackage[inline]{enumitem}

\usepackage{siunitx}
\DeclareSIUnit{\dBm}{dBm}

\usepackage{graphicx}
\usepackage[labelformat=simple]{subcaption}

\usepackage{epstopdf}
\usepackage{cuted}
\usepackage{multirow,booktabs}
\usepackage{diagbox}
\usepackage{makecell}
\usepackage{hhline}
\usepackage{array} 
\newcolumntype{x}{!{\vrule width 2px}}
\newcolumntype{y}{!{\vrule width 1.5px}}

\usepackage{optidef}

\usepackage[shortcuts,acronym,automake]{glossaries}
\input{glossary}

\usepackage{tcolorbox}
\tcbuselibrary{many}
\newtheorem{theorem}{Theorem}
\newtheorem{lemma}{Lemma}

\usepackage[norelsize,linesnumbered,ruled]{algorithm2e}
\SetKwRepeat{Do}{do}{while}
\makeatletter
\newcommand{\removelatexerror} {\let\@latex@error\@gobble}
\makeatother

\newcommand{\superscript}[1]{^{\mathrm{#1}}}
\newcommand{\subscript}[1]{_{\mathrm{#1}}}

\usepackage[normalem]{ulem}
\newcommand{\revise}[2]{{\color{red}\sout{#1}}{\color{blue}#2}}

\usepackage[textsize=tiny,colorinlistoftodos]{todonotes}
\makeatletter
\define@key{todonotes}{bh}[]{
	\setkeys{todonotes}{author=\textbf{Bin}, color=lime!30}}%
\define@key{todonotes}{yz}[]{
	\setkeys{todonotes}{author=\textbf{Yao}, color=blue!30}}%
\define@key{todonotes}{wc}[]{
	\setkeys{todonotes}{author=\textbf{Wenwen}, color=green!30}}%
\makeatother

\newif\ifreviewmode
\reviewmodetrue 
\reviewmodefalse

\ifreviewmode
\else
  \renewcommand{\todo}[1]{} 
  \renewcommand{\revise}[2]{#2} 
\fi

\newcommand\bob{\subscript{Bob}}
\newcommand\eve{\subscript{Eve}}
\newcommand\lf{\subscript{LF}}

\begin{document}

\title{Physical Layer Deception in OFDM Systems}

\author{
	
	\IEEEauthorblockN{
        Wenwen~Chen\IEEEauthorrefmark{1},
		Bin~Han\IEEEauthorrefmark{1},
		Yao~Zhu\IEEEauthorrefmark{2},
		Anke~Schmeink\IEEEauthorrefmark{2},
and~Hans~D.~Schotten\IEEEauthorrefmark{1}\IEEEauthorrefmark{3}
	}
	
	\IEEEauthorblockA{
		\IEEEauthorrefmark{1}RPTU Kaiserslautern-Landau, 
		\IEEEauthorrefmark{1}RWTH Aachen University,
		\IEEEauthorrefmark{3}German Research Center for Artificial Intelligence (DFKI)
	}
}

\bstctlcite{IEEEexample:BSTcontrol}

\maketitle

\begin{abstract}
 As a promising technology, \ac{pls} enhances security by leveraging the physical characteristics of communication channels. However, \revise{the conventional \ac{pls} approach leads to a considerable disparity in the effort legitimate users need to secure data compared to eavesdroppers.}{it commonly takes the legitimate user more effort to secure its data, compared to that required by the eavesdropper to intercept the communication.} To address this imbalance, we propose a \ac{pld} framework, which applies random deceptive ciphering combined with \ac{ofdm} to \revise{defend against eavesdropping proactively.}{deceive eavesdroppers with falsified information, preventing them from wiretapping.} While ensuring the same level of confidentiality as traditional \ac{pls} methods, the \ac{pld} approach additionally introduces a deception mechanism, \revise{}{which remains effective} even when the eavesdropper has the same knowledge about the transmitter as the legitimate receiver. Through detailed theoretical analysis and numerical simulations, we prove the superiority of our method over the conventional \ac{pls} approach.
\end{abstract}


\IEEEpeerreviewmaketitle

\glsresetall

\vspace{-3mm}
\section{Introduction}\label{sec:introduction}

\Ac{pls} has gained significant attention as a rising area of interest in wireless systems. Unlike traditional cryptographic methods, \ac{pls} exploits the characteristics of the physical channel and provides an additional layer of protection against eavesdropping. As an effective complement to traditional methods, \ac{pls} is becoming increasingly crucial in contemporary wireless networks \cite{HFA2019classification}.

Although \revise{most research}{most existing works} on \ac{pls} focus on infinite blocklength codes, it is crucial \revise{and necessary to exploit}{to consider} \ac{pls} on \ac{fbl} due to the future trend of \ac{urllc} \cite{she2021tutorial}, where the data \revise{packet only consist of a small number of bits}{packets are length-constrained} to support extremely reliable transmission with minimal latency. To access \ac{pls} performance with \ac{fbl}, the authors in \cite{YSP+2019wiretap} establish the bounds for the achievable security rate considering a specified leakage probability and error probability. Efforts such \revise{like}{as} \cite{liu2023energy, li2023joint, liu2023predictive} have been made to explore the \ac{fbl} regime for \ac{pls}. Furthermore, the authors in \cite{yang2019wiretap} investigate the maximal secrecy rate over a wiretap channel and its tightest bounds for discrete memoryless and Gaussian channels. The authors in \cite{wang2022achieving} maximize the secrecy rate under the covertness constraint by maintaining the confidential signal's signal-to-noise ratio below a certain threshold in the wiretap channel, preventing eavesdroppers from detecting the transmission. The interplay between reliability and security is studied in \cite{oh2023joint}, where the joint secure-reliability performance is improved by optimizing the allocation of transmission resources. In \cite{zhu2023trade}, the idea of trading reliability for security is introduced to describe the trade-off between security and reliability in \ac{pls} for short-packet transmissions.

However, the passive nature of \ac{pls} results in a notable imbalance between the legitimate users and the eavesdroppers, as the eavesdroppers can always attempt to wiretap with little risk of being detected%
, \revise{while}{whereas} legitimate users must take more precautions to secure data%
. To \revise{make up for this shortcoming}{address this limitation}, active defense methods should be \revise{introduced}{integrated} to \ac{pls}, such as deception technologies, which aim to confuse and distract potential eavesdroppers by generating false data or environments, thereby securing the real information. The principles of deception were initially introduced by \emph{Mitnick} \cite{mitnick2003art} in the field of social engineering and then adapted into defensive strategies, which were called \emph{honeypots} and then expanded to a wider range of deception technologies \cite{fraunholz2018demystifying}. However, in the physical layer of wireless systems, deception technologies are still in the early stages of development. In \cite{he2022proactive}, the spatial diversity of \ac{mimo} is \revise{utilized}{exploited} to lure an eavesdropper into a trap area where the fake messages are received. The authors in \cite{qi2024adversarial} design a \ac{gan} to generate waveforms that disrupt the eavesdropper's recognition model.

We proposed a novel framework \revise{of}{for} \ac{pld} in \cite{HZS+2023nonorthogonal} where \ac{nom} was applied to enhance security. This framework was the first to integrate \ac{pls} with deception technologies. We \revise{}{jointly} optimized the encryption rate and the power allocation to achieve high secure reliability and effective deception. We further improve the optimization problem in \cite{chen2025physical}, where we maximized the effective deception rate under the constraint of \ac{lfp} instead of directly combining the secrecy performance and deception performance. Additionally, we detailed the system model with both activated and deactivated ciphering and provided a comprehensive reception error model in different scenarios.
\revise{}{Although \ac{nom} improves security through the superposition of ciphertext and key, it faces limitations in practical applications. Compared to orthogonal schemes, \ac{nom} introduces additional decoding complexity due to the \ac{sic} \cite{islam2016power}. Furthermore, our previous work on optimizing key length imposed strict requirements on cipher design. Given these challenges, adopting \ac{ofdm} presents an attractive alternative, which is compatible with conventional wireless standards and frees the receiver from the \ac{sic} operation.}
\revise{In this work,}{Therefore, in this paper} we extend our previous work and investigate the performance of \ac{ofdm}-based \ac{pld}. Instead of optimizing the deception rate \revise{by setting the constraint of \ac{lfp}}{under a low \ac{lfp} constraint}, \revise{we consider the constraint of throughput so that a high deception rate can be attained while ensuring the security and efficiency of the transmission.}{we introduce a throughput constraint to achieve a high deception rate while maintaining both transmission security and efficiency.} By jointly optimizing the channel coding rates of ciphertext and key, \revise{we achieve a high deception rate while maintaining \ac{lfp} as low as in the conventional \ac{pls} method.}{the proposed framework attains a high deception rate while preserving an \ac{lfp} comparable to that of conventional \ac{pls} methods.}

The \revise{remaining part}{remainder} of this paper is organized as follows. \revise{We begin with setting up the system model and optimization problem in Sec. \ref{sec:problem}}{In Sec. \ref{sec:problem}, we establish the system model and formulate the optimization problem}. Afterward, we present our theoretical analyses and our optimization algorithm in Sec. \ref{sec:approach}. In Sec. \ref{sec:evaluation}, \revise{our approach is numerically validated and evaluated}{we validate and evaluate the approach through numerical simulations}. Finally, we \revise{conclude our paper and provide outlooks}{conclude the paper and outline potential directions for future research} in Sec. \ref{sec:conclusion}.

\section{Problem Setup}\label{sec:problem}
\subsection{System Model}
We consider an end-to-end communication system where the information source \emph{Alice} sends messages to the receiver \emph{Bob} over wireless channel $h\bob$ with gain $z\bob=\left\vert h\bob\right\vert^2$. At the same time, an eavesdropper \emph{Eve} \revise{an eavesdropping channel listens to \emph{Alice} over}{attempts to intercept the messages through} the eavesdropping channel $h\eve$ with gain $z\eve=\left\vert h\eve\right\vert^2$. With proper beamforming, \emph{Alice} can keep $h\bob$ statistically superior to $h\eve$, which is \revise{a necessary condition of \ac{pls} feasibility}{a prerequisite for the feasibility of \ac{pls}}. Our \revise{}{proposed} framework is illustrated in Fig. ~\ref{fig:alice_model}, where \emph{Alice} applies a two-stage encoder followed by \ac{ofdm}-based waveforming.
\begin{figure}[!htpb]
	\centering
    \vspace{-3mm}
	\begin{subfigure}[t]{\linewidth}
		\centering
        \frame{\includegraphics[width=0.8\linewidth]{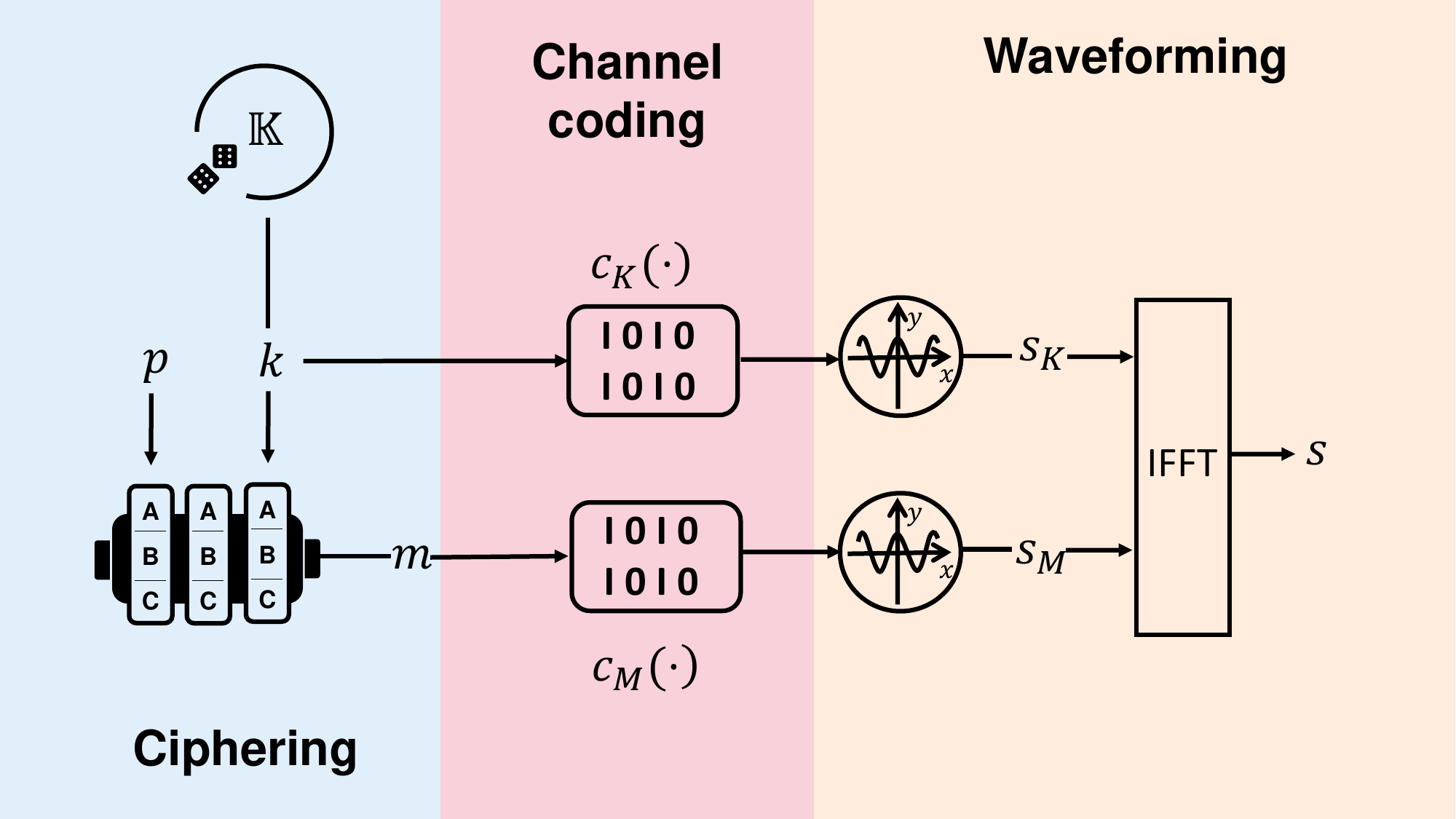}}
		\subcaption{}
		\label{subfig:alice_model_deception_active}
	\end{subfigure}\\
	\begin{subfigure}[t]{\linewidth}
		\centering
		\frame{\includegraphics[width=0.8\linewidth]{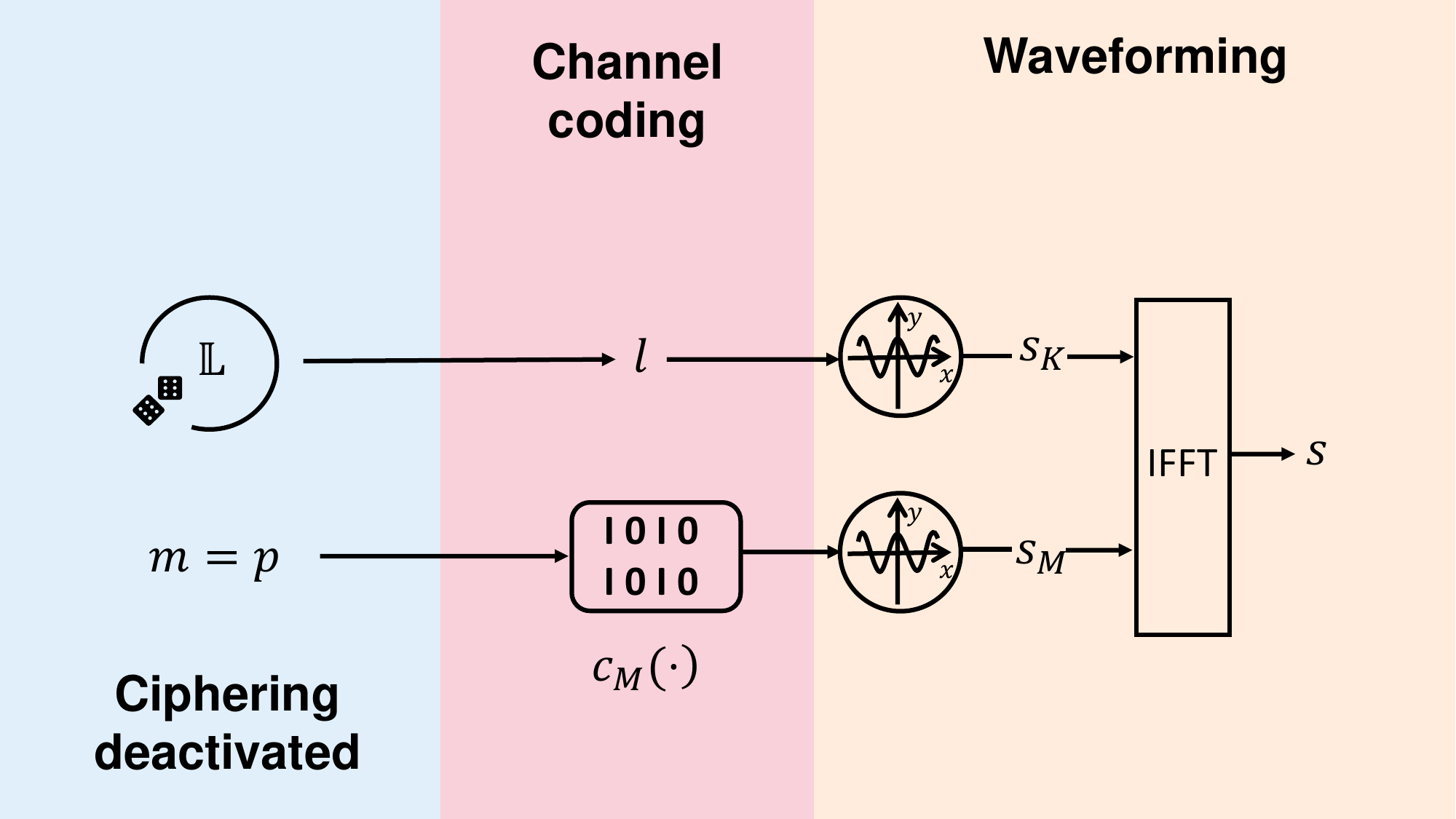}}
		\subcaption{}
		\label{subfig:alice_model_deception_inactive}
	\end{subfigure}
    \vspace{-2mm}
	\caption{The transmitting scheme of \emph{Alice}, with deceptive ciphering \subref{subfig:alice_model_deception_active} activated and \subref{subfig:alice_model_deception_inactive} deactivated, respectively.}
	\label{fig:alice_model}
    \vspace{-2mm}
\end{figure}

In this framework, the cipherer can be activated and deactivated by option. When activated, the $d\subscript{P}$-bit plaintext $p$ is encrypted into a $d\subscript{M}$-bit ciphertext using a $d\subscript{K}$-bit key $k$:
\begin{equation}
	m=f(p,k)\in\mathbb{M},\quad\forall (p,k)\in\left(\mathbb{P}\times\mathbb{K}\right),
\end{equation}
where $\mathbb{M}\subseteq \{0,1\}^{d\subscript{M}}$, $\mathbb{P}\subseteq \{0,1\}^{d\subscript{P}}$, and $\mathbb{K}\subseteq \{0,1\}^{d\subscript{K}}$ are the feasible sets of ciphertext codes, plaintext codes, and keys, respectively.
On the other hand, given the chosen key $k$, the plaintext can be decrypted from the ciphertext \revise{by}{using the inverse function}
$p=f^{-1}(m,k)$. 
Especially, codebooks must be designed \revise{to ensure}{such that the set of ciphertexts $\mathbb{M}$ is a subset of the plaintext set $\mathbb{P}$}: $\mathbb{M}\subseteq\mathbb{P}$.
\revise{and that}{Furthermore, for} $\forall\left(m,k,k'\right)\in\left(\mathbb{M}\times\mathbb{K}^2\right)$, it \revise{holds}{must hold that} $f^{-1}(m,k')\vert_{k'\neq k}\neq f^{-1}(m,k)$, \revise{}{ensuring that decryption with an incorrect key produces a result distinct from the correctly decrypted plaintext.}

The second stage involves channel coding, where error correction redundancies are added to both $m$ and $k$. The two output codewords are then individually modulated using \ac{ofdm}. On the receiver side, for both $i\in\{\text{Bob}, \text{Eve}\}$, \revise{it holds}{the received signal can be expressed as}
$r_i=s_i*h_i+w_i$,
where $s_i$ denotes the power-normalized baseband signal to transmit, $r_i$ is the received baseband signal, $h_i$ is the channel coefficient, and $w_i$ represents the equivalent baseband noise at receiver $i$.

On the other hand, when the cipherer is deactivated, \revise{}{ no encryption is performed}. Thus, the plaintext $p$ is directly \revise{inherited}{treated} as the ciphertext, i.e., $m=p$. Meanwhile, instead of \revise{}{using} a valid ciphering key $k\in\mathbb{K}$, a randomly generated ``litter'' sequence $l\in\mathbb{L}$ is used to derive $s\subscript{K}$. Particularly, the set of litter codes $\mathbb{L}\subseteq\{0,1\}^n$ shall fulfill
\begin{equation}
	\vspace{-2mm}
	\not\exists \{k,l\}\in\mathbb{K}\times\mathbb{L}: D\subscript{Hamm}(c\subscript{K}(k),l)\leqslant D\subscript{max},
\end{equation}
where $D\subscript{Hamm}(x,y)$ is the Hamming distance between $x$ and $y$, and $D\subscript{max}$ is the \revise{maximal distance of a received codeword from the codebook for the channel decoder}{maximum allowable distance for the channel decoder} $c^{-1}\subscript{K}$ to correct errors. The waveforming stage remains unchanged, following the same process as when the deceptive cipherer is activated.

Challenging the worst case where \emph{Eve} has maximum knowledge of this framework, we assume that the tuple $\left(\mathbb{P}, \mathbb{M}, \mathbb{K}, f, f^{-1}\right)$, as well as the modulation and channel coding schemes, are all \emph{common knowledge} shared among \emph{Alice}, \emph{Bob}, and \emph{Eve}. Assuming that both \emph{Bob} and \emph{Eve} have perfect knowledge of their own channels so that ideal channel equalization is achieved.

\subsection{Error Model}\label{subsec:error_model}
When the deceptive ciphering is activated, decoding both $m$ and $k$ can result in one of three possible outcomes:
\begin{enumerate}
\item \emph{Success}: When the bit errors fall within the error correction capability of the channel decoder, the data is retrieved.
\item \emph{Erasure}: If the bit errors surpass the receiver's error correction capability but \revise{not}{remain within} its error detection capability, the receiver will recognize and report an erasure.
\item \emph{Error}: If the bit errors exceed the error detection capability, the receiver will incorrectly decode the data, resulting in an undetected packet error.

\end{enumerate}

Practically, if \emph{Alice} is properly configured to encode both $m$ and $k$ with sufficient redundancy and transmit with adequate power, undetected error is unlikely to occur. Thus, there are three possible deciphering outcomes, as represented in Tab. \ref{tab:err_model}.
\begin{enumerate}
    \item \emph{Perception}: If both $m$ and $k$ are successfully decoded, the plaintext $p$ is correctly perceived by the receiver. 

    \item \emph{Loss}: If $m$ is erased, the receiver is unable to decrypt \revise{}{the message}, regardless \revise{its reception of $k$}{of whether $k$ is successfully decoded}, resulting in the loss of the $p$.

    \item \emph{Deception}: If the ciphering mechanism is randomly activated on selected messages (e.g., the most confidential ones), the deception can occur \revise{when the receiver successfully decodes $m$ but have $k$ erased.}{when $m$ is successfully decoded but $k$ is erased.} In this scenario, the receiver, unaware of whether the cipherer is active, cannot determine if the issue is caused by a transmission error or if the cipherer is inactive (meaning no $k$ but a random $l$ was transmitted). \revise{If the receiver mistakenly assumes the former as the latter}{If the receiver incorrectly assumes the issue to be an inactive cipherer}, it will interpret the ciphertext $m$ as unciphered plaintext, leading to a successful deception.
\end{enumerate}

\vspace{-3mm}
\begin{table}[!htpb]
    \centering
    \begin{tabular}{ c  c  c  c |}
        \multicolumn{2}{c}{}	&	\multicolumn{2}{c}{\textbf{Ciphertext}}\\
        \multicolumn{2}{c}{}	&	\textit{Success}	&	\multicolumn{1}{c}{\textit{Erasure}}	\\\hhline{*{2}~*{2}-}
        &\multicolumn{1}{r|}{\textit{Success}}	&	\multicolumn{1}{c|}{\cellcolor[gray]{0.9}Perception}	& {\cellcolor[gray]{0.7}} 	\\\hhline{*{2}~|->{\arrayrulecolor[gray]{0.7}}->{\arrayrulecolor{black}}|}
        \multirow{-2}{*}{\rotatebox{90}{\textbf{Key}}} &\multicolumn{1}{r|}{\textit{Erasure}}	& \multicolumn{1}{c|}{\cellcolor[gray]{0.5}Deception}	& \multirow{-2}{*}{Loss}{\cellcolor[gray]{0.7}}	\\\hhline{*{2}~|*{2}-|}
    \end{tabular}
    \caption{Reception error model of the proposed approach with random cipherer activation.}
    \label{tab:err_model}
\end{table}

\vspace{-5mm}
\subsection{Performance Metrices}
Conventional \ac{pls} approaches, which primarily operate in the \ac{ibl} regime, often rely on secrecy capacity to evaluate security performance. However, in the \ac{fbl} regime, the conventional notion of channel capacity is no longer applicable, as error-free transmission is rarely attainable~\cite{YSP+2019wiretap}. \revise{}{Thus, secrecy capacity is not valid to evaluate the secrecy performance of short-packet communication systems \cite{wang2019secure}.}
 \revise{}{To evaluate the reliable-secure performance for a single transmission,} we introduce the \ac{lfp}, \revise{}{defined as} $\varepsilon\lf=1-(1-\varepsilon\bob)\varepsilon\eve$. This metric represents the probability that the plaintext is either correctly perceived by the eavesdropper \emph{Eve} or not perceived by the legitimate user \emph{Bob} \cite{zhu2023trade}.
Here, $\varepsilon\bob$ and $\varepsilon\eve$ are the non-perception probabilities of \emph{Bob} and \emph{Eve}, respectively. Notating $\varepsilon_{i,j}$ as the erasure probability of receiver $i\in\{\text{Bob}, \text{Eve}\}$ \revise{at}{when} decoding the message component $j\in\{\text{M}, \text{K}\}$, \revise{we have}{the overall erasure probability for receiver $i$ is given by:} $\varepsilon_i=1-(1-\varepsilon_{i,\text{M}})(1-\varepsilon_{i,\text{K}})$. Therefore, the \ac{lfp} can be calculated as:
$\varepsilon\lf=1-\left (1-\varepsilon\subscript{Bob,M} \right) \left (1-\varepsilon\subscript{Bob,K}\right)\left[ 1-\left (1-\varepsilon\subscript{Eve,M} \right) \left (1-\varepsilon\subscript{Eve,K} \right) \right]$.
Additionally, to evaluate the performance of deceiving eavesdroppers, we define the effective deception rate as the probability that not \emph{Bob} but only \emph{Eve} is deceived, i.e.,
	$R\subscript{d}=\left[1-\left(1-\varepsilon\subscript{Bob,M}\right)\varepsilon\subscript{Bob,K}\right](1-\varepsilon\subscript{Eve,M})\varepsilon\subscript{Eve,K}$.
According to~\cite{PPV2010channel}, the error probability $\varepsilon_{i,j}$ with a given packet size $d_j$ can be written as
    $\varepsilon_{i,j} = Q \left(\sqrt{\frac{n_j}{V(\gamma_{i})}}(\mathcal{C}(\gamma_{i})-\frac{d_j}{n_j}) \ln{2}\right)$,
where  $Q(x)=\frac{1}{\sqrt{2 \pi}} \int_{x}^{\infty} e^{-t^2/2} dt$ is the Q-function in statistic, $\mathcal{C}(\gamma_i)=\log_2(1+\gamma_i)$ is the Shannon capacity, ${V}(\gamma_i)=1-\frac{1}{{{\left( 1+{{\gamma_i }} \right)}^{2}}}$ is the channel dispersion with $\lambda_i=\frac{z_iP}{\sigma^2}$.

\subsection{Strategy Optimization}
\revise{For convenience of analysis, we assume the subcarriers are allocated to the ciphertext and the key equally.}{For analytical convenience, we assume that the subcarriers are equally allocated between the ciphertext and the key. } The cyclic prefix duration is set to 0 and the bandwidth is normalized to 1.%
\revise{}{It is worth noting that at low SNR (i.e. below $\SI{0}{\dB}$), the capacities of BPSK and $M$-QAM approach the Shannon limit \cite{gobel2010information}. Therefore, without loss of generality, we continue to use the Shannon formula for capacity calculation and adopt BPSK as the modulation scheme}. 
The throughput is defined as
    $T=\left(1-\varepsilon\subscript{\lf}\right)\left(\frac{d\subscript{M}}{n\subscript{M}+n\subscript{K}}\right)$.
\revise{In order to pursue}{To achieve} secure and efficient transmission, we \revise{}{aim to} maximize the effective deception rate while maintaining high throughput. To simplify the encryptor design, The ciphertext length $d\subscript{M}$ and the key length $d\subscript{K}$ are fixed , while their coding rates are adjustable to ensure transmission efficiency. The optimization problem can be formulated as:
\begin{maxi!}
	{n\subscript{M}, n\subscript{K}\in\mathbb{Z}^+ }{R\subscript{d}} {\label{Obj fun: Deception rate}}{\label{Problem_orig}}{}
	\addConstraint{\varepsilon\subscript{Bob,M}\leqslant \varepsilon\superscript{th}\subscript{Bob,M}, \varepsilon\subscript{Eve,M}\leqslant \varepsilon\superscript{th}\subscript{Eve,M} \label{con:err_message_thresholds}}
	\addConstraint{\varepsilon\subscript{Bob,K}\leqslant \varepsilon\superscript{th}\subscript{Bob,K}, \varepsilon\subscript{Eve,K}\geqslant \varepsilon\superscript{th}\subscript{Eve,K} \label{con:err_key_thresholds}}
    \addConstraint{T \geqslant T\superscript{th} 
    \label{con:throughput_threshold}}
\end{maxi!}

\vspace{-3mm}
\section{Proposed Approach}\label{sec:approach}
\vspace{-1mm}
\subsection{Analyses}\label{subsec:analyses}
The original problem (\ref{Obj fun: Deception rate}) is challenging to solve due to the non-convexity of $R\subscript{d}$. Therefore, we reformulate it into a simpler but equivalent version.  Based on our analytical insights, we demonstrate that the objective function exhibits partial convexity with respect to each optimization variable.

We first relax $n\subscript{M}$ and $n\subscript{K}$ from integers to real values:
\begin{maxi!}
	{n\subscript{M},  n\subscript{K}\in \mathbb{R}^+ }{R\subscript{d}} {\label{Problem_SCA}}{}
    \addConstraint{\eqref{con:err_message_thresholds} - \eqref{con:throughput_threshold}\nonumber}
\end{maxi!}
Subsequently, the original problem can be decomposed with:
\begin{lemma}
\label{lemma:approximation}
For a given $\left(\hat{n}\superscript{(q)}\subscript{M}, \hat{n}\superscript{(q)}\subscript{K}\right)$, $R\subscript{d}$ is lower-bounded by an approximation $\hat{R}\subscript{d} \left(n\subscript{M}, n\subscript{K} \vert \hat{n}\superscript{(q)}\subscript{M}, \hat{n}\superscript{(q)}\subscript{K}\right)$, i.e.,
\begin{align}
    R\subscript{d}(n\subscript{M},n\subscript{K})&=\left[1-\left(1-\varepsilon\subscript{Bob,M}\right)\varepsilon\subscript{Bob,K}\right](1-\varepsilon\subscript{Eve,M})\varepsilon\subscript{Eve,K}\nonumber \\
    &\geqslant \left[1-\left(1-\hat{\varepsilon}\subscript{Bob,M}(\hat{n}\subscript{M}\superscript{(q)},\hat{n}\subscript{K}\superscript{(q)})\right)\varepsilon\subscript{Bob,K}\right] \nonumber
     \\
   &\cdot \left((1-\varepsilon\subscript{Eve,M})\hat{\varepsilon}\subscript{Eve,K}(\hat{n}\subscript{M}\superscript{(q)}, \hat{n}\subscript{K}\superscript{(q)})\right) \label{Rd-approxi}\\
   & \triangleq \hat{R_d}\left(n\subscript{M},n\subscript{K}\vert \hat{n}\subscript{M}\superscript{(q)},\hat{n}\subscript{K}\superscript{(q)}\right)\nonumber
\end{align}
where
    $\hat{\varepsilon}\subscript{Bob,M}(\hat{n}\subscript{M}\superscript{(q)},\hat{n}\subscript{K}\superscript{(q)})=1\!-\!b(-\hat {\omega })e^{-a(-\hat {\omega })\omega\subscript{M} }\!-\!c(-\hat {\omega })$
and
     $\hat{\varepsilon}\subscript{Eve,K}(\hat{n}\subscript{M}\superscript{(q)},\hat{n}\subscript{K}\superscript{(q)})=1\!-\!b(-\hat{\omega})e^{-a(-\hat {\omega })\omega\subscript{K} }\!-\!c(-\hat {\omega })$.
\end{lemma}

\begin{proof}
	Omitted due to length limitation, see \cite{chen2024physical}.
\end{proof}

In Eq.~\eqref{Rd-approxi}, $R\subscript{d}$ reaches the lower-bound $\hat{R}\subscript{d}$ at the point $\left(\hat{n}\subscript{M}\superscript{(q)},\hat{n}\subscript{K}\superscript{(q)}\right)$, which inspires us to utilize the \ac{mm} algorithm combining with the \ac{bcd} \cite{tseng2001convergence} method to solve the problem. Thus, we first decompose the problem in each $t\superscript{th}$ iteration by fixing $n\subscript{M}$. The corresponding problem is given by:

\begin{maxi!}
	{n\subscript{K}\in\mathbb{R}^+}{\left.\hat{R}\subscript{d}^{(t)}\left(n\subscript{K}\left\vert \hat{n}\subscript{M}^{(q)},\hat{n}\subscript{K}^{(q)}\right.\right)\right\vert_{n\subscript{M}=n\subscript{M}^{(t)}}} {\label{Problem_BCD_n_K_origi}}{}{}
	\addConstraint{\eqref{con:err_message_thresholds} - \eqref{con:throughput_threshold}.\nonumber}
\end{maxi!}

Next, we leverage the \ac{fp} \cite{shen2018fractional} to further decouple the problem. In this way, Problem (\ref{Problem_BCD_n_K_origi}) is equivalent to the following problem:
\begin{maxi!}
    {n\subscript{K}\in\mathbb{R}^+,y}{\left.\hat{f}^{(t)}\left(n\subscript{K},y\left\vert\hat{n}\subscript{M}^{(q)},\hat{n}\subscript{K}^{(q)}\right.\right)\right\vert_{n\subscript{M}=n\subscript{M}^{(t)}}} {\label{Problem_BCD_n_K}}{}{}
    \addConstraint{\eqref{con:err_message_thresholds} - \eqref{con:throughput_threshold},\nonumber}
\end{maxi!}
where
\begin{equation}\label{eq:f_t}
\begin{split}
&\hat{f}^{(t)}\left(n\subscript{K},y\vert \hat{n}\subscript{M}^{(q)},\hat{n}\subscript{K}^{(q)}\right)\\
&=2y\sqrt{\left[1-\left(1-\hat{\varepsilon}\subscript{Bob,M}^{(t)}(\hat{n}\subscript{M}^{(q)},\hat{n}\subscript{K}^{(q)})\right)\varepsilon\subscript{Bob,K}\right]} \\
&-y^2\frac{1}{\left((1-\varepsilon\subscript{Eve,M}^{(t)})\hat{\varepsilon}\subscript{Eve,K} (\hat{n}\subscript{M}^{(q)}, \hat{n}\subscript{K}^{(q)})\right)}. 
 \end{split}
\end{equation}
\vspace{-4mm}
\begin{theorem}
\label{theorem:f_t_concave}    
Eq.~\eqref{eq:f_t} is concave in $n\subscript{K}$ for fixed $y$.
\end{theorem}
\vspace{-3mm}
\begin{proof}
	Omitted due to length limitation, see \cite{chen2024physical}.
\end{proof}
\vspace{-1mm}
In Problem \eqref{Problem_BCD_n_K}, we introduce an additional variable $y$ and construct the quadratic transform which is concave for fixed $y$ and $n\subscript{K}$ separately. Thus, we can solve this problem via the \ac{bcd} method and find the optimal solution $n\subscript{K}\superscript{o}$ efficiently.

On the other hand, we have the second decomposed problem in the $t\superscript{th}$ iteration by fixing $n\subscript{K}$:
\begin{maxi!}
	{n\subscript{M}\in\mathbb{R}^+}{\left.\hat{R}\subscript{d}^{(t)}\left(n\subscript{M}\left\vert \hat{n}\subscript{M}^{(q)},\hat{n}\subscript{K}^{(q)}\right.\right)\right\vert_{n\subscript{K}=n\subscript{K}^{(t)}}} {\label{Problem_BCD_n_M_origi}}{}{}
    \addConstraint{\eqref{con:err_message_thresholds} - \eqref{con:throughput_threshold}.\nonumber}
\end{maxi!}
Similarly, Problem (\ref{Problem_BCD_n_M_origi}) can be reformulated as:
    \begin{maxi!}
	{n\subscript{M}\in\mathbb{R}^+,y}{\left.\hat{g}^{(t)}\left(n\subscript{M},y\left\vert\hat{n}\subscript{M}^{(q)},\hat{n}\subscript{K}^{(q)}\right.\right)\right\vert_{n\subscript{K}=n\subscript{K}^{(t)}}} {\label{Problem_gt_n_M}}{}{}
    \addConstraint{\eqref{con:err_message_thresholds} - \eqref{con:throughput_threshold},\nonumber}
\end{maxi!}
\begin{equation}\label{eq:g_t}
\begin{split}
&\hat{g}^{(t)}\left(n\subscript{M},y\vert \hat{n}\subscript{M}^{(q)},\hat{n}\subscript{K}^{(q)}\right)\\
&=2y\sqrt{\left((1-\varepsilon\subscript{Eve,M})\hat{\varepsilon}\subscript{Eve,K}\superscript{(t)} (\hat{n}\subscript{M}^{(q)}, \hat{n}\subscript{K}^{(q)})\right)} \\
&-y^2\frac{1}{\left[1-\left(1-\hat{\varepsilon}\subscript{Bob,M}(\hat{n}\subscript{M}^{(q)},\hat{n}\subscript{K}^{(q)})\right)\varepsilon\subscript{Bob,K}^{(t)}\right]}. 
 \end{split}
\end{equation}
\vspace{-3mm}
\begin{theorem}
\label{theorem:g_t_concave}
Eq.~\eqref{eq:g_t} is concave in $n\subscript{M}$ for fixed $y$.
\end{theorem}
\vspace{-3mm}
\begin{proof}
	Omitted due to length limitation, see \cite{chen2024physical}.
\end{proof}
\vspace{-2mm}
Therefore, we can also solve Problem ($\ref{Problem_gt_n_M}$) via \ac{bcd} approach and obtain the optimal solution $n\subscript{M}\superscript{o}$.

\subsection{Optimization Algorithm}\label{subsec:algorithm}

\begin{algorithm}[!hbp]
	\scriptsize
	\SetAlgoLined
	Input: $\mu\subscript{BCD},\mu\subscript{MM}, \mu \subscript{FP}, T,Q,I,J, P, n\subscript{M}, n\subscript{K}$ \\
	Initialize: $t=1, q=1, i=1, j=1, n\subscript{M}\superscript{(0)}=n\subscript{M}\superscript{init}, n\subscript{K}\superscript{(0)}=n\subscript{K}\superscript{init} ,
	\hat{n}\subscript{M}\superscript{(0)}=\hat{n}\subscript{M}\superscript{init}, \hat{n}\subscript{K}\superscript{(0)}=\hat{n}\subscript{K}\superscript{init} ,
	R\subscript{d}\superscript{(0)}=-\infty$ \\
	\Do{$\frac{\hat{R}^{(t)}\subscript{d}\left(\hat{n}^{(q)}\subscript{M},\hat{n}^{(q)}\subscript{K}\right)-\hat{R}^{(t)}\subscript{d}\left(\hat{n}^{(q-1)}\subscript{M},\hat{n}^{(q-1)}\subscript{K}\right)}{\hat{R}^{(t)}\subscript{d}\left(\hat{n}^{(q-1)}\subscript{M},\hat{n}^{(q-1)}\subscript{K}\right)}> \mu\subscript{MM}$
		}{\eIf{
			$q\leqslant Q$
		}{
			$t\leftarrow 1$ (reset index $t$)\\
			$\hat{f}^{(t)}:= \hat{f}\left(\hat{n}^{(q)}\subscript{M},\hat{n}^{(q)}\subscript{K}\right)$, 
                $\hat{g}\superscript{(t)}:= \hat{g}\left(\hat{n}^{(q)}\subscript{M},\hat{n}^{(q)}\subscript{K}\right)$ \\
			\Do{
				$\frac{\hat{R}\subscript{d}^{(t)}-\hat{R}\subscript{d}^{(t-1)}}{\hat{R}\subscript{d}^{(t-1)}}> \mu\subscript{BCD}$
				}{
				\eIf{
					$t\leqslant T$
				}{\Do{$\frac{\hat{f}^{(i)}-\hat{f}^{(i-1)}}{\hat{f}^{(i-1)}}>\mu \subscript{FP}$}{$i\leftarrow 1 (\text{reset index} \;i)$ \\
    \eIf{$i \leqslant I$}{$y^{(i)}= \frac{\sqrt{\left[1-\left(1-\hat{\varepsilon}\subscript{Bob,M}^{(t-1)}\right)\varepsilon\subscript{Bob,K}(n\subscript{K}^{(i-1)})\right]}}{\left(1-\varepsilon\subscript{Eve,M}^{(t-1)})\hat{\varepsilon}\subscript{Eve,K}(n\subscript{K}^{(i-1)}) \right)}$ \\
                        $n\subscript{K}^{(i)}\leftarrow \arg \underset{n\subscript{K}}{\max}\; \hat{f}\left(y^{(i)},n\subscript{M}^{(t-1)}\right)$ \\
                        $i\leftarrow i+1$} {\textbf{break}}} 
                        $n\subscript{K}^{(t)} \leftarrow n\subscript{K}^{(i)}$

                        \Do{$\frac{\hat{g}^{(i)}-\hat{g}^{(i-1)}}{\hat{g}^{(i-1)}}>\mu \subscript{FP}$}{$j \leftarrow 1 (\text{reset index} \; j)$ \\
                        \eIf{$j \leqslant J$}{$y^{(j)}= \frac{\sqrt{\left(1-\varepsilon\subscript{Eve,M}^{(j-1)})\hat{\varepsilon}\subscript{Eve,K}(n\subscript{K}^{(t)}) \right)}}{\left[1-\left(1-\hat{\varepsilon}\subscript{Bob,M}^{(j-1)}\right)\varepsilon\subscript{Bob,K}(n\subscript{K}^{(t)})\right]}$ \\
                        $n\subscript{M}^{(j)}\leftarrow \arg \underset{n\subscript{M}}{\max}\; \hat{g}\left(y^{(j)},n\subscript{K}^{(t)}\right)$ \\
                        $j\leftarrow j+1$} {\textbf{break}}
                        }
                        $n\subscript{M}^{(t)}\leftarrow n\subscript{M}^{(j)}$
    
					$\hat{R}\subscript{d}^{(t)}\leftarrow \hat{R}\subscript{d}\left(n\subscript{M}^{(t)},n\subscript{K}^{(t)}\right)$,
					$t \leftarrow t+1$
				}{
				\textbf{break}
				}
			}
			$\hat{n}\subscript{M}^{(q)}\leftarrow n\subscript{M}^{(t)}$, 
			$\hat{n}\subscript{K}^{(q)}\leftarrow n\subscript{K}^{(t)}$,
			$q\leftarrow q+1$
		}{
		\textbf{break}
		}
	}
	$n\subscript{M}\superscript{*}\leftarrow \arg \underset{n \in \left\{\left\lfloor n\subscript{M}^{(q)} \right\rfloor, \left\lceil n\subscript{M}^{(q)} \right\rceil\right\}}{\max} R\subscript{d}\left(n\subscript{K}^{(q)}\right)$ \\
	$n\subscript{K}\superscript{*}\leftarrow \arg \underset{n \in \left\{\left\lfloor n\subscript{K}^{(q)} \right\rfloor, \left\lceil n\subscript{K}^{(q)} \right\rceil\right\}}{\max} R\subscript{d}\left(n\subscript{M}^{(q)}\right)$ \\
	\KwRet{$\left(n\subscript{M}\superscript{*}, n\subscript{K}\superscript{*}\right)$}
\caption{The proposed \ac{mm}-\ac{bcd}-\ac{fp} framework}
\label{alg:framework}
\end{algorithm}
Based on the above analyses, we propose an algorithm with three layers of iterations. In the $q\superscript{th}$ iteration, we approximate $\hat{R}\subscript{d}^{(q)}:=\hat{R}\subscript{d}\left(\hat{n}\subscript{M}\superscript{(q)},\hat{n}\subscript{K}\superscript{(q)}\right)$. Next, we fix the value of $n\subscript{M}$ and $n\subscript{K}$ respectively in $t\superscript{th}$ iteration. In this way, we can solve the single-variable problem via \ac{fp} approach, which is equivalent to solving problem \eqref{Problem_BCD_n_K} and \eqref{Problem_gt_n_M}. In particular, in the inner $i\superscript{th}$ iteration for fixed $n\subscript{M}^{(t)}$, the optimal $y^{(i)}$ can be found in a closed form for fixed $n\subscript{K}^{(i-1)}$. And $n\subscript{K}^{(i)}$ can be updated by solving the reformulated convex optimization problem. Furthermore, we can use the same way to calculate the optimal $n\subscript{M}^{(t)}$. Then the local point 
will be assigned to $\left(n\subscript{M}^{(q+1)}, n\subscript{K}^{(q+1)}\right)$ for the next
iteration. The process repeats until the relative error is less than the threshold or the maximum number of iterations is achieved.

Specifically, the initial values of $n\subscript{M}$ and $n\subscript{K}$ should be feasible for Problem \eqref{Problem_orig}. Besides, since both $n\subscript{M}$ and $n\subscript{K}$ must be integers, the optimal solution must be determined by comparing the integer neighbors of $n\subscript{M}^{*}$ and $n\subscript{K}^{*}$. \revise{}{This approach to solve Problem~\eqref{Problem_SCA} is described in Algorithm~\revise{}{\ref{alg:framework}}. The method can achieve near-optimal solutions with a complexity of $\mathcal{O}\left(\phi\left(8N^3\right)\right)$, where $N$ denotes the number of variables in Problem~\eqref{Problem_SCA} and $\phi(\cdot)$ signifies the number of iterations based on the accuracy of the solution.}

\vspace{-1mm}
\section{Numerical Evaluation}\label{sec:evaluation}
To validate our analyses and evaluate the proposed approach, we conducted a series of numerical experiments. The general parameters of the simulation setup are presented in Tab. \ref{tab:sim_setup}, while task-specific parameters will be explained later.

First, we set the transmission power $P=\SI{5}{\milli \watt}$ for both ciphertext and key, under the condition that $z\bob=\SI{0}{\dB}$ and $z\eve=\SI{-10}{\dB}$. We calculated $R_d$ in the region $(n\subscript{M}, n\subscript{K}) \in [16, 128]\times \ [16, 128]$ with $T\superscript{th}=0.1$ bps. The result in Fig. \ref{fig:deception_rate-surface} illustrates the concavity of the deception rate surface in the feasible region, which is constrained by (\ref{con:err_message_thresholds}-\ref{con:err_key_thresholds}) and highlighted with greater opacity compared to the rest. However, the behavior related to convexity or concavity beyond this region seems to be more complex.

To verify the effectiveness of the proposed \ac{bcd} algorithm, we conducted Monte-Carlo simulations, where we set $T\superscript{th}=0.1$ bps with transmission power $P=\SI{5}{\milli\watt}$, $z\subscript{\bob}=\SI{-5}{\dB}$, $z\subscript{\eve}=\SI{-15}{\dB}$. Fig. \ref{fig:bcd} illustrates the search path with $d_M=16$ bits and $d_M=24$ bits and proves that the \ac{bcd} algorithm converges at the optimum.

\begin{table}[!htpb]
	\centering
	\caption{Simulation setup}
	\label{tab:sim_setup}
	\begin{tabular}{m{1.2cm} m{1.5cm} m{4.3cm}}
		\toprule[2px]
		\textbf{Parameter} 		&\textbf{Value} 		&\textbf{Remark}\\
		\midrule[1px]
		$\sigma^2$					&\SI{1}{\milli\watt}			&Noise power\\
  \rowcolor{gray!30}
		$B$								 &\SI{1}{\hertz}			&Normalized to unity bandwidth\\
		$d\subscript{M}$								 &16 bits							 &Length of ciphertext\\
  \rowcolor{gray!30}
        $d\subscript{K}$								 &16 bits							 &Length of key\\

		$\varepsilon\subscript{Bob,M}\superscript{th}$ $\varepsilon\subscript{Bob,K}\superscript{th}$ $\varepsilon\subscript{Eve,M}\superscript{th}$ $\varepsilon\subscript{Eve,K}\superscript{th}$		&0.5							 &Thresholds in constraints \eqref{con:err_message_thresholds}--\eqref{con:err_key_thresholds}\\
        \rowcolor{gray!30}
		$\xi\subscript{MM}$							&$2\times 10^{-16}$	&\ac{mm} convergence threshold\\
        $\xi\subscript{BCD}$							&$2\times 10^{-16}$	&\ac{bcd} convergence threshold\\
         \rowcolor{gray!30} $\xi\subscript{FP}$							&$2\times 10^{-16}$	&\ac{fp} convergence threshold \\
		$K$								&100							&Maximal number of iterations in \ac{bcd}\\
		\bottomrule[2px]
	\end{tabular}
\end{table}

To evaluate the secrecy and deception performance of our proposed methods, we calculated $\varepsilon_{LF}$ and $R_d$ under varying eavesdropping channel gain $z\subscript{\eve}$. In this experiment, we set $P=\SI{5}{\milli \watt}$, $z\subscript{\bob}=\SI{0}{\dB}$, and $T\superscript{th}=0.05$. We also tested the conventional \ac{pls} method as a baseline, which minimizes $\varepsilon\subscript{\lf}$ with respect to $n\subscript{M}$ without deceptive ciphering ($d\subscript{K}=0$). The results are represented in Fig. \ref{fig:sensitivity-zeve}, which demonstrate that better eavesdropping channel condition enhances the deception performance. Although the $\varepsilon\subscript{\lf}$ of our method increases with the growth of $z\subscript{\eve}$, it still remains at a very low value and performs closely to the baseline.

Next, we set $z\subscript{\eve}=\SI{-15}{\dB}$, $z\subscript{\bob}=\SI{0}{\dB}$, $T\subscript{th}=0.05$ bps to test the performance w.r.t. the transmission power. The results are shown in Fig. \ref{fig:sensitivity-p}. which indicates that the deception performance benefits from higher transmission power, while the $\varepsilon\subscript{\lf}$ rises logarithmically slowly as the transmission power increases.

Fig. \ref{fig:sensitivity-raw_packet_rate} shows the sensitivity of $\varepsilon_\text{LF}$ to the raw packet rate, which is tested under $z\subscript{\bob}=\SI{0}{\dB}$, $z\subscript{\eve}=\SI{-10}{\dB}$, and $T\superscript{th}=0.05$ bps. Compared with the conventional \ac{pls} method, our \ac{pld} method benefits from a lower raw data rate. As the raw packet rate increases, the deception performance degrades, while the leakage failure probability rises.

The outcome of a comprehensive benchmark test is depicted in Fig \ref{fig:benchmark}, where we combined $z\subscript{\eve}$ and transmission power. We kept the setup $T\superscript{th}=0.05$. The deception rate rises with better eavesdropping channel conditions and higher transmission power. Regarding the $\varepsilon\subscript{\lf}$, the \ac{pld} method performs closely to the conventional \ac{pls} method. The $\varepsilon\subscript{\lf}$ gets larger as the $z\subscript{\eve}$ and transmission power increase.

\begin{figure}[!htpb]
	\centering
	\vspace{-5mm}
	\includegraphics[width=\linewidth]{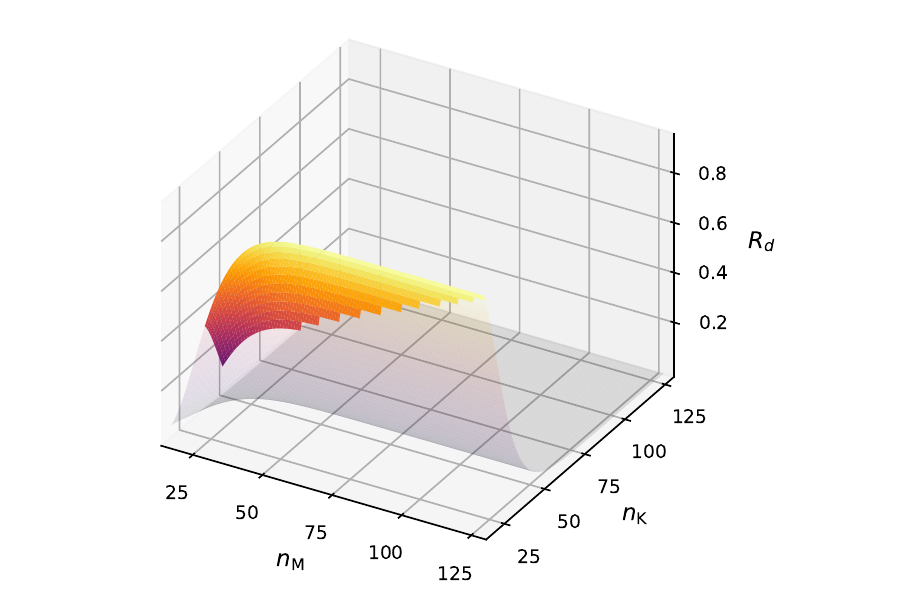}
	\vspace{-8mm}
	\caption{Deception rate with $T\superscript{th}\lf=0.1$ bps. }
	\label{fig:deception_rate-surface}
\end{figure}
\begin{figure}[!htpb]
	\centering
	\includegraphics[width=\linewidth]{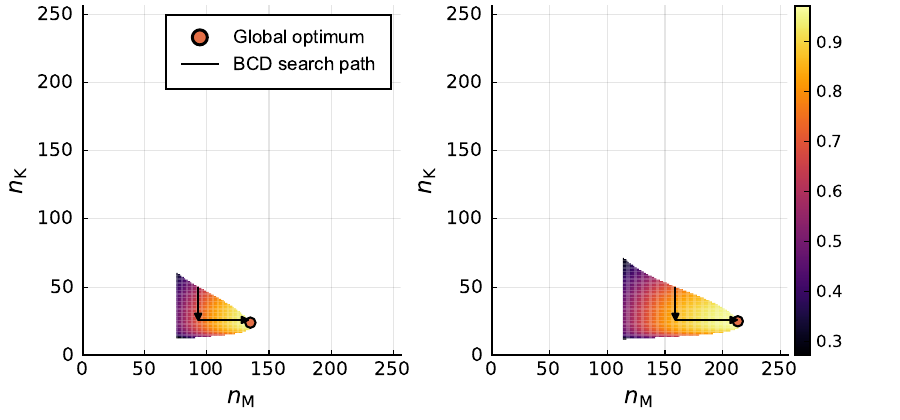}
	\caption{The $R_d$ surface and the search path with $d\subscript{M}=16$ bits (left) and $d\subscript{M}=24$ bits (right).}
	\label{fig:bcd}
\end{figure}

\section{Conclusion}\label{sec:conclusion}
In this work, we have investigated the performance of our proposed \ac{pld} framework with \ac{ofdm}. By jointly optimizing the coding rate of the ciphertext and the key, we maximized the effective deception rate while maintaining a specified throughput constraint, thereby ensuring both secure and efficient communication. We have proved the convexity of the objective function and proposed an efficient algorithm to solve the related optimization problem. The comprehensive numeral simulation results have demonstrated that our approach introduced high deception rate without compromising security compared with the conventional \ac{pls} method.


\begin{figure}[!htpb]
	\centering
	\includegraphics[width=.8\linewidth]{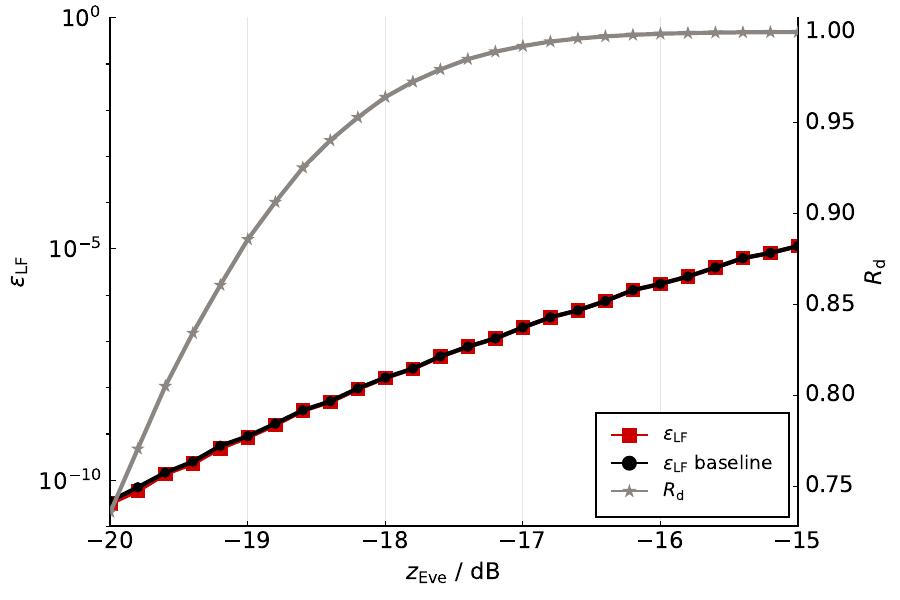}
	\vspace{-3mm}
	\caption{The impact of $z\subscript{\eve}$.}
	\label{fig:sensitivity-zeve}
\end{figure}
\begin{figure}[!htpb]
	\centering
	\includegraphics[width=.8\linewidth]{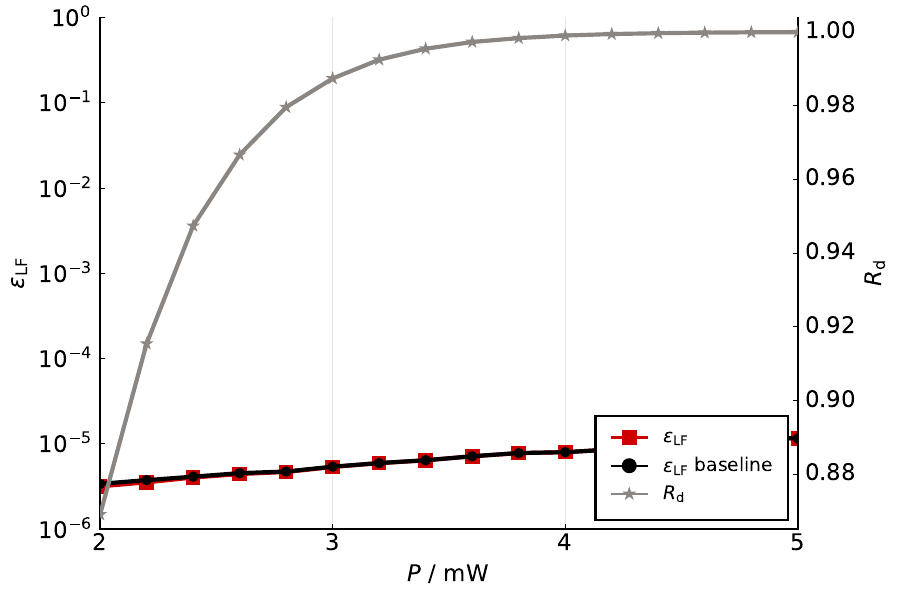}
	\vspace{-3mm}
	\caption{The impact of $P$.}
	\label{fig:sensitivity-p}
\end{figure}
\begin{figure}[!htpb]
	\centering
	\vspace{-3mm}
	\includegraphics[width=.8\linewidth]{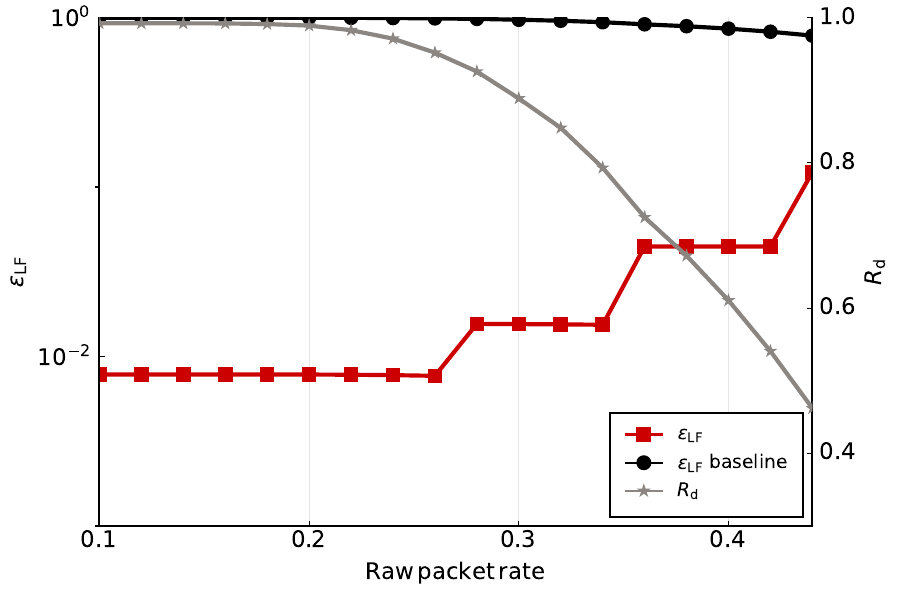}
	\vspace{-3mm}
	\caption{The impact of raw packet rate.}
	\label{fig:sensitivity-raw_packet_rate}
\end{figure}
\begin{figure}[!htpb]
	\centering
	\vspace{-3mm}
	\includegraphics[width=.9\linewidth]{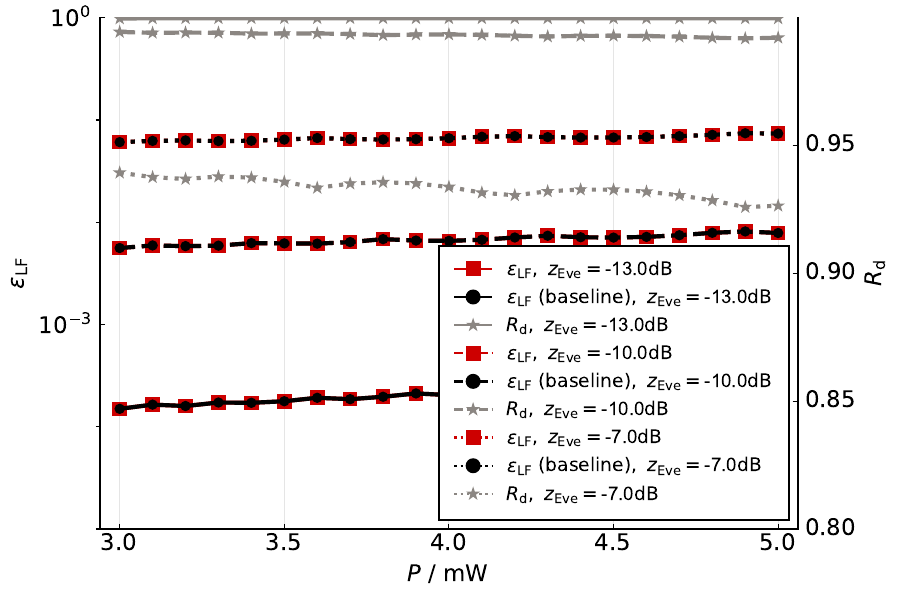}
	\vspace{-3mm}
	\caption{Benchmark results}
	\label{fig:benchmark}
\end{figure}

\section*{Acknowledgment}
This work is supported the German Federal Ministry of Education and Research in the programme of "Souverän. Digital. Vernetzt." joint projects 6G-ANNA (16KISK105/16KISK097) and Open6GHub (16KISK003K/16KISK004/16KISK012).


\bibliographystyle{IEEEtran}
\bibliography{references}

\clearpage
\appendices
\section{Proof of Lemma~\ref{lemma:approximation}}\label{app:lemma_approximation}
\begin{proof}
	Eq.~\eqref{Rd-approxi} can be obtained from the approximation of the Q-function according to lemma 3 in \cite{zhu2023trade}. We introduce and auxiliary function $\omega=\sqrt{\frac{n}{V(\lambda)}}\left(\mathcal{C}(\lambda)-\frac{d}{n}\right)\ln{2}$. For a given $\hat{w}\in \mathcal{R}$, Q-function is bounded by 
		$1\!-\!b(-\hat {\omega })e^{-a(-\hat {\omega })\omega }\!-\!c(-\hat {\omega })\leqslant Q(\omega)\leqslant b(\hat {\omega })e^{-a(\hat {\omega })\omega }+c(\hat {\omega })$,
	where
		$a(\hat {\omega })=\max \left\{{\frac {e^{-\frac {(\hat {\omega })^{2}}{2}}}{\sqrt {2\pi }Q(\hat {\omega })},\hat {\omega }}\right\}>0$,
		$b(\hat {\omega })=\frac {1}{\sqrt {2\pi }\hat a}e^{\hat a\hat {\omega }-\frac {(\hat {\omega })^{2}}{2}}>0$, 
	and
		$c(\hat {\omega })=Q(\hat {\omega })-\hat be^{-\hat a\hat {\omega }}$. 
	The equality is taken for $\omega=\hat{\omega}$.
\end{proof}

\section{Proof of Theorem~\ref{theorem:f_t_concave}}\label{app:theorem_f_t_concave}
\begin{proof}
    To prove the convexity of $\hat{f}^{(t)}$, we first investigate the monotonicity of $\varepsilon_{i,j}$ with respect to $n_j$. In particular, we have
    \begin{equation}
        \frac {\partial {\varepsilon_{i,j}}}{\partial{n_j}}=\frac {\partial {\varepsilon_{i,j}}}{\partial{w_{i,j}}}\frac {\partial {w_{i,j}}}{\partial{n_j}} \leqslant 0,
    \end{equation}
    where
    \begin{equation}
        \frac {\partial {\varepsilon_{i,j}}}{\partial{w_{i,j}}}=\frac {\partial \left ({\int ^{\infty} _{w_{i,j}} \frac {1}{\sqrt {2\pi }}e^{-\frac {t^{2}}{2}}dt }\right)}{\partial w_{i,j}}=-\frac {1}{\sqrt {2\pi }}e^{-\frac {w_{i,j}^{2}}{2}} < 0,
    \end{equation}
    \begin{equation}
        \frac {\partial {w_{i,j}}}{\partial {n_j}}=\frac {1}{2}n_j^{-\frac {1}{2}}V_{i,j}^{-\frac {1}{2}}\mathcal {C}_{i,j}\ln {2}+\frac {1}{2}n_j^{-\frac {3}{2}}V_{i,j}^{-\frac {1}{2}}d_j\ln {2}\geqslant 0.
    \end{equation}
    Thus, $\varepsilon_{i,j}$ is monotonically decreasing in $n_j$. Then, we further investigate the convexity of $\varepsilon_{i,j}$ with respect to $n_j$, we have
    \begin{equation}
    \begin{split}
        \frac{\partial^2 \varepsilon_{i,j}}{\partial n_{j}^2}=\frac{\partial^2 \varepsilon_{i,j}}{\partial w_{i,j}^2}\underbrace{\left(\frac {\partial {w_{i,j}}}{\partial{n_j}}\right)^2}_{\geqslant 0}+\underbrace{\frac{\partial \varepsilon_{i,j}}{\partial w_{i,j}}}_{< 0}\frac{\partial^2 w_{i,j}}{\partial n_j^2}\geqslant 0,
    \end{split}
    \end{equation}
    where
    \begin{equation}
        \frac {\partial ^{2}{\varepsilon_{i,j} }}{\partial {w_{i,j}^{2}}}=\frac {w_{i,j}}{\sqrt {2\pi }}e^{-\frac {w_{i,j}^{2}}{2}} \geqslant 0,
    \end{equation}
    \begin{equation}
        \frac {\partial ^{2}{w_{i,j}}}{\partial {n_j}^{2}}=-\frac {1}{4}n_j^{-\frac {3}{2}}V_{i,j}^{-\frac {1}{2}}\mathcal{C}_{i,j}\ln {2}\!-\!\frac {3}{4}n_j^{-\frac {5}{2}}V_{i,j}^{-\frac {1}{2}}d_j \ln {2}\leqslant 0.
    \end{equation}
    Therefore, $\varepsilon_{i,j}$ is convex in $n_j$. We can further prove the concavity of the first term in $\hat{f}^{(t)}$, where $\left[1-\left(1-\hat{\varepsilon}\subscript{Bob,M}^{(t)}(\hat{n}\subscript{M}^{(q)},\hat{n}\subscript{K}^{(q)})\right)\varepsilon\subscript{Bob,K}\right]$ is concave. Since the square-root function is concave and increasing, $2y\sqrt{\left[1-\left(1-\hat{\varepsilon}\subscript{Bob,M}^{(t)}(\hat{n}\subscript{M}^{(q)},\hat{n}\subscript{K}^{(q)})\right)\varepsilon\subscript{Bob,K}\right]}$ is concave.
    
    Next, we prove the concavity of $\hat{\varepsilon}_{i,j}$. The first derivative of $\hat{\varepsilon}_{i,j}$ is:
    \begin{equation}
        \frac {\partial {\hat{\varepsilon}_{i,j}}}{\partial{n_j}}=\frac {\partial {\hat{\varepsilon}_{i,j}}}{\partial{w_{i,j}}}\underbrace{\frac {\partial {w_{i,j}}}{\partial{n_j}}}_{\geqslant 0}\geqslant 0,
    \end{equation}
    where
    \begin{equation}
        \frac {\partial {\hat{\varepsilon}_{i,j}}}{\partial{w_{i,j}}}=a(-\hat{w})b(-\hat{w})e^{-a(-\hat{w})w_{i,j}}>0.
    \end{equation}
    Then, we further investigate the concavity of $\hat{\varepsilon}_{i,j}$ with respect to $n_j$, we have
    \begin{equation}
        \frac{\partial^2 \hat{\varepsilon}_{i,j}}{\partial n_{j}^2}=\frac{\partial^2 \hat{\varepsilon}_{i,j}}{\partial w_{i,j}^2}\underbrace{\left(\frac {\partial {w_{i,j}}}{\partial{n_j}}\right)^2}_{\geqslant 0}+\underbrace{\frac{\partial \hat{\varepsilon}_{i,j}}{\partial w_{i,j}}}_{> 0}\underbrace{\frac{\partial^2 w_{i,j}}{\partial n_j^2}}_{\leqslant 0}\leqslant 0,
    \end{equation}
    where
    \begin{equation}
        \frac{\partial^2 \hat{\varepsilon}_{i,j}}{\partial w_{i,j}^2}=-a^2(-\hat{w})b(-\hat{w})e^{-a(-\hat{w})w}<0.
    \end{equation}
    Therefore, the second term of $\hat{f}^{(t)}$ is concave with respect to $n\subscript{K}$. Hence, $\hat{f}^{(t)}$ is concave. It is also trivial to show that all the constraints are either convex or linear, i.e., the feasible set of Problem (\ref{Problem_BCD_n_K}) is convex. Since the objective function to be maximized is concave and its feasible set is convex, Problem (\ref{Problem_BCD_n_K}) is a convex problem.
\end{proof}

\section{Proof of Theorem~\ref{theorem:g_t_concave}}
\label{app:theorem_g_t_concave}
\begin{proof}
According to the proof in Appendix \ref{app:theorem_f_t_concave}, $\varepsilon\subscript{Eve,M}$ is convex and $\hat{\varepsilon}\subscript{Bob,M}^{(t)}$ is concave with respect to $n\subscript{M}$. Thus, $\hat{g}^{(t)}$ is concave in $n\subscript{M}$.
\end{proof}





\end{document}

%% file: glossary.tex
\makeglossaries
\newacronym{awgn}{AWGN}{additive white Gaussian noise}
\newacronym{bcd}{BCD}{block coordinate descent}
\newacronym{cp}{CP}{control plane}
\newacronym{csi}{CSI}{channel state information}
\newacronym{csit}{CSIT}{channel state information at transmitter}
\newacronym{dft-s-ofdm}{DFT-s-OFDM}{Discrete Fourier Transform-spread-OFDM}
\newacronym{fbl}{FBL}{finite blocklength}
\newacronym{gan}{GAN}{generative adversarial network}
\newacronym{ibl}{IBL}{infinite blocklength}
\newacronym{lfp}{LFP}{leakage-failure probability}
\newacronym{mimo}{MIMO}{multi-input multi-output}
\newacronym{mm}{MM}{Majorize-Minimization}
\newacronym{noma}{NOMA}{non-orthogonal multi-access}
\newacronym{nom}{NOM}{non-orthogonal multiplexing}
\newacronym{ofdm}{OFDM}{orthogonal frequency-division multiplexing}
\newacronym{ofdma}{OFDMA}{orthogonal frequency-division multiple access}
\newacronym{oma}{OMA}{orthogonal multiple access}
\newacronym{papr}{PAPR}{Peak-to-Average Power Ratio}
\newacronym{per}{PER}{packet error rate}
\newacronym{phy}{PHY}{physical}
\newacronym{pld}{PLD}{physical layer deception}
\newacronym{pls}{PLS}{physical layer security}
\newacronym{prb}{PRB}{physical resource block}
\newacronym{sic}{SIC}{successive interference cancellation}
\newacronym{sinr}{SINR}{signal-to-interference-and-noise ratio}
\newacronym{snr}{SNR}{signal-to-noise ratio}
\newacronym{tdma}{TDMA}{time-division multiple access}
\newacronym{up}{UP}{user plane}
\newacronym{urllc}{URLLC}{ultra-reliable low-latency communication}
\newacronym{fp}{FP}{fractional programming}

%% file: deceptive_pls.bbl
\begin{thebibliography}{10}
\providecommand{\url}[1]{#1}
\csname url@samestyle\endcsname
\providecommand{\newblock}{\relax}
\providecommand{\bibinfo}[2]{#2}
\providecommand{\BIBentrySTDinterwordspacing}{\spaceskip=0pt\relax}
\providecommand{\BIBentryALTinterwordstretchfactor}{4}
\providecommand{\BIBentryALTinterwordspacing}{\spaceskip=\fontdimen2\font plus
\BIBentryALTinterwordstretchfactor\fontdimen3\font minus
  \fontdimen4\font\relax}
\providecommand{\BIBforeignlanguage}[2]{{%
\expandafter\ifx\csname l@#1\endcsname\relax
\typeout{** WARNING: IEEEtran.bst: No hyphenation pattern has been}%
\typeout{** loaded for the language `#1'. Using the pattern for}%
\typeout{** the default language instead.}%
\else
\language=\csname l@#1\endcsname
\fi
#2}}
\providecommand{\BIBdecl}{\relax}
\BIBdecl

\bibitem{HFA2019classification}
J.~M. Hamamreh, H.~M. Furqan, and H.~Arslan, ``Classifications and applications
  of physical layer security techniques for confidentiality: {A} comprehensive
  survey,'' \emph{IEEE Commun. Surv. Tutor.}, vol.~21, no.~2, pp. 1773--1828,
  2019.

\bibitem{she2021tutorial}
C.~She, C.~Sun, Z.~Gu \emph{et~al.}, ``A tutorial on ultrareliable and
  low-latency communications in 6g: Integrating domain knowledge into deep
  learning,'' \emph{Proc. IEEE}, vol. 109, no.~3, pp. 204--246, 2021.

\bibitem{YSP+2019wiretap}
W.~Yang, R.~F. Schaefer, and H.~V. Poor, ``Wiretap channels: {Nonasymptotic}
  fundamental limits,'' \emph{IEEE Trans. Inf. Theory}, vol.~65, no.~7, pp.
  4069--4093, 2019.

\bibitem{liu2023energy}
B.~Liu, P.~Zhu, J.~Li \emph{et~al.}, ``Energy-efficient optimization in
  distributed massive {MIMO} systems for slicing {eMBB} and {URLLC} services,''
  \emph{IEEE Trans. Veh. Technol.}, vol.~72, no.~8, pp. 10\,473--10\,487, 2023.

\bibitem{li2023joint}
K.~Li, P.~Zhu, Y.~Wang \emph{et~al.}, ``Joint uplink and downlink resource
  allocation toward energy-efficient transmission for {URLLC},'' \emph{IEEE J.
  Sel. Areas Commun.}, vol.~41, no.~7, pp. 2176--2192, 2023.

\bibitem{liu2023predictive}
C.~Liu, S.~Li, W.~Yuan \emph{et~al.}, ``Predictive precoder design for
  otfs-enabled urllc: A deep learning approach,'' \emph{IEEE J. Sel. Areas
  Commun.}, vol.~41, no.~7, pp. 2245--2260, 2023.

\bibitem{yang2019wiretap}
W.~Yang, R.~F. Schaefer, and H.~V. Poor, ``Wiretap channels: Nonasymptotic
  fundamental limits,'' \emph{IEEE Trans. Inf. Theory}, vol.~65, no.~7, pp.
  4069--4093, 2019.

\bibitem{wang2022achieving}
C.~Wang, Z.~Li, H.~Zhang \emph{et~al.}, ``Achieving covertness and security in
  broadcast channels with finite blocklength,'' \emph{IEEE Trans. Wireless
  Commun.}, vol.~21, no.~9, pp. 7624--7640, 2022.

\bibitem{oh2023joint}
M.~Oh, J.~Park, and J.~Choi, ``Joint optimization for secure and reliable
  communications in finite blocklength regime,'' \emph{IEEE Trans. Wireless
  Commun.}, vol.~22, no.~12, pp. 9457--9472, 2023.

\bibitem{zhu2023trade}
Y.~Zhu, X.~Yuan, Y.~Hu \emph{et~al.}, ``Trade reliability for security:
  Leakage-failure probability minimization for machine-type communications in
  urllc,'' \emph{IEEE J. Sel. Areas Commun.}, vol.~41, no.~7, pp. 2123--2137,
  2023.

\bibitem{mitnick2003art}
K.~D. Mitnick and W.~L. Simon, \emph{The Art of Deception: Controlling the
  Human Element of Security}.\hskip 1em plus 0.5em minus 0.4em\relax John Wiley
  \& Sons, 2003.

\bibitem{fraunholz2018demystifying}
D.~Fraunholz, S.~D. Anton, C.~Lipps \emph{et~al.}, ``Demystifying deception
  technology: A survey,'' 2018, [Online]. Available: arXiv:1804.06196.

\bibitem{he2022proactive}
Q.~He, S.~Fang, T.~Wang \emph{et~al.}, ``Proactive anti-eavesdropping with trap
  deployment in wireless networks,'' \emph{IEEE Trans. Dependable Secure
  Comput.}, vol.~20, no.~1, pp. 637--649, 2022.

\bibitem{qi2024adversarial}
P.~Qi, Y.~Meng, S.~Zheng \emph{et~al.}, ``Adversarial defense embedded waveform
  design for reliable communication in the physical layer,'' \emph{IEEE
  Internet Things J.}, 2024.

\bibitem{HZS+2023nonorthogonal}
B.~Han, Y.~Zhu, A.~Schmeink \emph{et~al.}, ``Non-orthogonal multiplexing in the
  {FBL} regime enhances physical layer security with deception,'' in \emph{2023
  IEEE SPAWC)}, 2023, pp. 211--215.

\bibitem{chen2025physical}
W.~Chen, B.~Han, Y.~Zhu \emph{et~al.}, ``Physical layer deception with
  non-orthogonal multiplexing,'' \emph{IEEE Trans. Wireless Commun.}, 2025.

\bibitem{islam2016power}
S.~R. Islam, N.~Avazov, O.~A. Dobre \emph{et~al.}, ``Power-domain
  non-orthogonal multiple access (noma) in 5g systems: Potentials and
  challenges,'' \emph{IEEE Commun. Surv. Tutor.}, vol.~19, no.~2, pp. 721--742,
  2016.

\bibitem{wang2019secure}
H.-M. Wang, Q.~Yang, Z.~Ding \emph{et~al.}, ``Secure short-packet
  communications for mission-critical iot applications,'' \emph{IEEE Trans.
  Wireless Commun.}, vol.~18, no.~5, pp. 2565--2578, 2019.

\bibitem{PPV2010channel}
Y.~Polyanskiy, H.~V. Poor, and S.~Verdu, ``Channel coding rate in the finite
  blocklength regime,'' \emph{IEEE Trans. Inf. Theory}, vol.~56, no.~5, pp.
  2307--2359, 2010.

\bibitem{gobel2010information}
B.~W. G{\"o}bel, ``Information-theoretic aspects of fiber-optic communication
  channels,'' Ph.D. dissertation, Technische Universit{\"a}t M{\"u}nchen, 2010.

\bibitem{chen2024physical}
W.~Chen, B.~Han, Y.~Zhu \emph{et~al.}, ``Physical layer deception in {OFDM}
  systems,'' 2024, [Online]. Available: arXiv:2411.03677.

\bibitem{tseng2001convergence}
P.~Tseng, ``Convergence of a block coordinate descent method for
  nondifferentiable minimization,'' \emph{J. Optim. Theory Appl.}, vol. 109,
  pp. 475--494, 2001.

\bibitem{shen2018fractional}
K.~Shen and W.~Yu, ``Fractional programming for communication systems—part i:
  Power control and beamforming,'' \emph{IEEE Trans. Signal Process.}, vol.~66,
  no.~10, pp. 2616--2630, 2018.

\end{thebibliography}
